\documentclass[pra,twocolumn,aps,epsfigure,showpacs]{revtex4}
\usepackage{epsfig}
\usepackage{graphics}
\usepackage{amsmath}
\usepackage{amsthm}

\newcommand{\nn}{\nonumber \\}
\newcommand{\bra}[1]{\langle{#1}|}
\newcommand{\ket}[1]{|{#1}\rangle}

\newcommand{\dg}{^\dagger}
\newcommand{\Tr}{\textrm{Tr}}
\newcommand{\be}{{\cal B}}
\newcommand{\p}{{\cal E}}
\newcommand{\s}{s}
\newtheorem{thm}{Theorem}
\newtheorem{lem}{Lemma}

\begin{document}

\title{The standard fair sampling assumption is not necessary to test local realism}

\author{Dominic W. Berry,$^{1,2}$ Hyunseok Jeong,$^{3,4}$  Magdalena Stobi\'nska,$^{3,5,6}$
and Timothy C. Ralph$^{3}$}
\affiliation{$^1$Centre for Quantum Computer Technology, Macquarie University,
Sydney, NSW 2109, Australia\\
$^2$Institute for Quantum Computing, University of Waterloo, Waterloo, ON N2L 3G1, Canada\\
$^3$Centre for Quantum Computer Technology, Department of Physics,
The University of Queensland, St Lucia, Qld 4072, Australia\\
$^4$Center for Subwavelength Optics, Department of Physics and Astronomy, Seoul National University, Seoul, 151-742, Korea\\
$^5$Institut f\"ur Optik, Information und Photonik,
  Universit\"at Erlangen-N\"urnberg, Staudtstr. 7/B2, 91058 Erlangen,
  Germany\\
$^6$Institute of Theoretical Physics and Astrophysics,
  University of Gda\'nsk, ul.~Wita Stwosza 57, 80-952 Gda\'nsk, Poland}

\begin{abstract}
Almost all Bell-inequality experiments to date have used
postselection, and therefore relied on the fair sampling assumption for their
interpretation. The standard form of the fair sampling assumption is that the
loss is independent of the measurement settings, so the ensemble of detected systems provides a
fair statistical sample of the total ensemble.
This is often assumed to be needed to interpret Bell inequality experiments as ruling out hidden-variable
theories. Here we show that it is not necessary; the loss can depend on measurement
settings, provided the detection efficiency factori{\s}es as a function of the
measurement settings and any hidden variable. 
This condition implies that Tsirelson's bound must be satisfied for
entangled states. On the other hand, we show that it is possible for
Tsirelson's bound to be violated while the CHSH-Bell inequality still
holds for unentangled states, and present an experimentally feasible
example.
\end{abstract}

\maketitle

\section{Introduction}
When quantum mechanics was first developed it was argued that it might only be
an approximation of an underlying classical ``hidden-variable'' theory \cite{epr}.
This was put on a testable basis by the development of Bell inequalities
\cite{bell,chsh,ch74}, which should be obeyed by any local hidden variable (LHV) theory.
Experimental violation \cite{freedman,aspect,shih,ou,rarity,tittel,weihs,kwiat,rowe,moehring, matsukevich} of Bell inequalities
provides strong evidence against LHV theories, but almost all
of these experiments (with the exceptions of those in Refs.\ \cite{rowe,matsukevich})
have the loophole that the violation of the Bell
inequality could, in principle, be caused by loss \cite{lo,pearle,garg,santos,eber,gisin}. The interpretation of these
experiments as ruling out LHV theories therefore relies on the
assumption that the sampling is ``fair'' \cite{chsh,fair}. For the
interpretation of these experiments it is therefore vital to establish what
constitutes fair sampling. Here we put fair sampling on a rigorous basis by
determining exactly what forms of loss can lead to violation of Bell
inequalities.

The standard form of the fair sampling assumption is that the
detection efficiency is independent of measurement settings \cite{chsh,fair}. Here we
find that assumption is unnecessary. The efficiency can depend on the
measurement settings, provided the efficiency factori{\s}es as a function of the
measurement settings and any hidden variable. Most experimental tests are of the 
Clauser, Horne, Shimony and Holt (CHSH) form of the Bell inequality \cite{chsh}, and
we therefore concentrate on the CHSH-Bell inequality in this work. Our condition
is both necessary and sufficient for the CHSH-Bell inequality to be satisfied
for LHV theories. An alternative sufficient condition was previously
found by Ref.\ \cite{shaf}. We also establish the necessary and sufficient
condition for the CHSH-Bell inequality to be satisfied for unentangled states,
and show that if the sampling is fair it will also prevent violation of
Tsirelson's bound \cite{tsirelson} with entangled states.

This result means that, in order to obtain violation of Tsirelson's bound \cite{tsirelson} with entangled states, but not of the CHSH-Bell inequality with unentangled states, it is necessary to examine the specific measurement, and not just the form of the loss. We present a scheme that violates Tsirelson's bound with entangled states, but not the CHSH-Bell inequality with unentangled states. The CHSH-Bell inequality is still violated with LHV theories, but the scheme can tolerate greater loss than the bound derived in Ref.\ \cite{garg}. Other methods of constructing Bell inequalities with greater resistance to loss have been proposed in Refs.\ \cite{eber,massar,durt,genovese,brun}.

This manuscript is organi{\s}ed as follows. First the fair sampling assumption is explained in more detail in Sec.\ \ref{sec:fsa}. General Bell inequalities are presented in Sec.\ \ref{sec:bell}, then postselection for local hidden variable theories is analysed in Sec.\ \ref{sec:lhv}. Postselection for quantum mechanics, including Tsirelson's bound, is analysed in Sec.\ \ref{sec:qm}. The use of postselection to enhance violation of Bell inequalities beyond Tsirelson's bound for entangled states is analysed in Sec.\ \ref{sec:tsi}. We conclude in Sec.\ \ref{sec:conc}.

\section{The fair sampling assumption}
\label{sec:fsa}
The first work to give a form of the fair sampling assumption was that of Clauser, Horne, Shimony and Holt \cite{chsh}. They assumed, in deriving their inequality, that the detection efficiency is independent of the measurement settings. Pearle expressed the assumption alternatively, that ``the data recorded [is] representative of the accepted data'' \cite{pearle}. That is, that the sample of detected pairs provides a fair statistical sample of all the pairs. These forms of the fair sampling assumption, which superficially appear different, are effectively equivalent.

For the sample to be completely fair, the probability of sampling a pair (i.e.\ the efficiency) needs to be independent of the pair. For this to be the case, the efficiency needs to be independent of any quantity that varies between the individual pairs. The only quantities that can vary between the individual pairs are the internal state (the quantum state or any hypothetical hidden variable) and the measurement setting that is used in measuring the pair. Therefore the requirement that the sample is fair is exactly equivalent to requiring the efficiency to be independent of the internal state and measurement setting.

This requirement is slightly stronger than the requirement given by Ref.\ \cite{chsh}, because it requires that the efficiency is also independent of the state of the pair. However, it is easily seen that the efficiency need not be independent of the state. This is because, if the efficiency is dependent on the state, but independent of the measurement setting, then the postselection simply changes the probabilities for the internal state, yielding a different postselected state. That is, the sampling may not be fair, but the loss is simply yielding a postselected state on which the sampling is fair.

This is taken advantage of in Procrustean entanglement concentration \cite{Bennett}, which has been demonstrated to enhance the violation of Bell inequalitites \cite{Kwiat01}. Procrustean entanglement concentration gives loss that depends on the state, but because this loss is independent of the measurement settings of the Bell measurement, it does not invalidate the Bell inequalities for unentangled states. Although the sampling is not entirely fair for the initial state, it is fair for the state produced by the Procrustean entanglement concentration.

It is important to note that the condition that the efficiency is independent of the measurement setting means that it must be \emph{completely} independent of the measurement setting. That is, it can not be a function of the measurement setting, so it is independent of the measurement setting for any internal state. On the other hand, if the efficiency is independent of the measurement setting for just one state, then it is easy to provide examples where the Bell inequality is violated with unentangled states or hidden variables. For a simple example, see Appendix \ref{sec:example}.

\section{General Bell inequalities}
\label{sec:bell}
In a general Bell inequality experiment with multiple parties, these parties each share one
component of a state $\rho$, and each performs one of a number of
different measurements. The measurement settings for party $k$ are denoted
$\gamma_k$, and the measurement results are denoted $s_k$. We denote the vectors of measurement settings and measurement results $\vec \gamma$ and $\vec s$.
One obtains a set of measurement probabilities
$p(\vec s|\vec \gamma,\rho)$, and can define a Bell quantity as a linear
combination of these probabilities. A Bell inequality is then an upper bound on
the value of this quantity for LHV theories. With loss, we denote the probability
of a successful measurement for settings $\vec \gamma$ by $\p(\vec \gamma,\rho)$. The postselected probabilities are then given by
\begin{equation}
\tilde p(\vec s|\vec \gamma,\rho) = p(\vec s|\vec \gamma,\rho)/\p(\vec \gamma,\rho).
\end{equation}
We do not consider complete loss, which would make the postselected
probabilities undefined.

In the specific case of CHSH-Bell inequalities, there are two parties, Alice and Bob, and each performs
one of two dichotomic measurements. Now using the notation $\alpha$ and $\beta$ for the measurement settings for Alice and Bob, respectively, we have $\alpha\in\{A,a\}$, $\beta\in\{B,b\}$,
and $s_1,s_2\in\{-,+\}$. The CHSH-Bell quantity without postselection is defined by
\begin{equation}
\be_0 := C(A,B)+C(A,b)+C(a,B)-C(a,b),
\end{equation}
where
\begin{align}
C(\alpha,\beta) &:= p(+,+|\alpha,\beta,\rho)-p(+,-|\alpha,\beta,\rho) \nonumber \\
& \quad -p(-,+|\alpha,\beta,\rho)+p(-,-|\alpha,\beta,\rho).
\end{align}
The postselected form of the CHSH-Bell quantity is
\begin{equation}
\label{eq:post}
\be := \frac{C(A,B)}{\p(A,B,\rho)}+\frac{C(a,B)}{\p(a,B,\rho)}
+\frac{C(A,b)}{\p(A,b,\rho)}-\frac{C(a,b)}{\p(a,b,\rho)}.
\end{equation}

\section{Local hidden-variable theories}
\label{sec:lhv}
First we prove that Bell inequalities must be satisfied for LHV theories provided
the detection efficiency factori{\s}es as a function of the measurement setting and
a function of the hidden variable. For a stochastic local hidden-variable theory,
probabilities are given as
\begin{equation}
p(\vec s|\vec \gamma,\rho) = \sum_x p(x|\rho) \prod_{k} p_k(s_k|\gamma_k,x),
\end{equation}
where $x$ is the hidden variable. The state $\rho$ simply controls the probability distribution for $x$. With loss, the probabilities $p(\vec s|\vec \gamma,\rho)$ can sum to less
than 1. Summing over the measurement results gives the efficiencies as
\begin{equation}
\p(\vec \gamma,\rho) = \sum_x p(x|\rho) \prod_{k} \p_k(\gamma_k,x),
\end{equation}
where $\p_k(\gamma_k,x)=\sum_s p(s|\gamma_k,x)$ is the single-party efficiency for party $k$ with
measurement setting $\gamma_k$ and
hidden variable~$x$. 

We also use the notation $\p_{k{\rm S}}(\gamma_k)$ for the efficiency due to measurement
setting $\gamma_k$ for party $k$, and $\p_{k{\rm H}}(x)$ for the efficiency due to the hidden variable for party $k$. With this notation defined, we can now state our condition rigorously.

\begin{thm}
When the efficiency for each party factori{\s}es as
\begin{align}
\label{eq:req}
\p_{k}(\gamma_k,x) &= \p_{k{\rm S}}(\gamma_k)\p_{k{\rm H}}(x),
\end{align}
the set of postselected probabilities that can be obtained $\{\tilde p(\vec s|\vec \gamma,\rho)\}$ is identical to that which can be obtained without postselection (provided no efficiency is zero).
\end{thm}

This result means that this form of loss does
not change the type of probability distributions that are possible with LHV
models, and in particular any Bell inequality must still be satisfied
with postselection.

\begin{proof}
The proof follows by showing that there exists a measurement scheme without postselection that yields the same probabilities as the $\tilde p(\vec s|\vec \alpha,\rho)$ for the postselected scheme.
The postselected probabilities may alternatively be given by
\begin{equation}
\tilde p(\vec s|\vec \gamma,\rho) = \sum_x \tilde p(x|\rho) \prod_k \tilde p(s_k|\gamma_k,x),
\end{equation}
where
\begin{align}
\tilde p(x|\rho) &= \frac{p(x|\rho)}{\p_{\rm R}(\rho)}\prod_k\p_{k{\rm H}}(x), \nn
\p_{\rm R}(\rho) &= \sum_x p(x|\rho) \prod_k \p_{k{\rm H}}(x), \nn
\tilde p(s_k|\gamma_k,x) &= {p(s_k|\gamma_k,x)}/{\p_k(\gamma_k,x)}.
\end{align}
This means that, by using a different LHV model with different probabilities,
we may obtain exactly the same probability distribution for the measurement
results as for the postselected case.
\end{proof}

This result shows that our condition \eqref{eq:req} is sufficient for Bell inequalities to hold with postselection. We also have a necessity proof in the case of CHSH-Bell inequalities. Before showing this result, we first show the general form of the CHSH-Bell inequality with postelection.

\begin{lem}
Any local hidden variable theory must satisfy
\begin{equation}\label{geneq}
|\be| \le 4-2\sum_x p(x|\rho)\min_{\alpha^x,\beta^x}
\frac{\p_1(\alpha^x,x)\p_2(\beta^x,x)}{\p(\alpha^x,\beta^x,\rho)}.
\end{equation}
For given $\p$ there exist probabilities $p_k(\pm|\gamma_k,x)$ that saturate this
inequality.
\end{lem}

Here we have used the same notation as for general Bell inequalities, except we have used $\gamma_1=\alpha$ and $\gamma_2=\beta$ for Alice and Bob's measurement settings, respectively. We use the superscript $x$ on $\alpha$ and $\beta$ to indicate that these are chosen as a function of the hidden variable $x$.

\begin{proof}
For a LHV theory, the postselected CHSH-Bell quantity can be rewritten as
\begin{align}
\label{sinz}
\be &= \sum_x p(x|\rho) \left\{\frac{\Delta_1(A,x)\Delta_2(B,x)}{\p(A,B,\rho)}
+\frac{\Delta_1(a,x)\Delta_2(B,x)}{\p(a,B,\rho)}\right. \nn & \quad\left.
+\frac{\Delta_1(A,x)\Delta_2(b,x)}{\p(A,b,\rho)}
-\frac{\Delta_1(a,x)\Delta_2(b,x)}{\p(a,b,\rho)}\right\},
\end{align}
where
\begin{equation}
\Delta_k(\gamma_k,x) = p_k(+|\gamma_k,x)-p_k(-|\gamma_k,x) \nn.
\end{equation}
Because the probabilities
are non-negative, we have $|\Delta_k(\gamma_k,x)|\le \p_k(\gamma_k,x)$. Considering
arbitrary measurements for a given form of loss, we can choose the
$p_k(\pm|\gamma_k,x)$ such that $\Delta_k(\gamma_k,x)$ takes any value in the range
$[-\p_k(\gamma_k,x),\p_k(\gamma_k,x)]$.

Because $\be$ is linear in each of the $\Delta_k(\gamma_k,x)$, it is maximi{\s}ed (or
minimi{\s}ed) by taking extreme values where $\Delta_k(\Gamma_k,x)=\pm\p_k(\gamma_k,x)$.
By appropriately choosing the signs, it is possible to make any one of the terms
in the curly brackets in Eq.\ \eqref{sinz} negative, and the rest positive.
Changing the sign preserves the parity of the number of negative terms, so it is
not possible to make all terms positive.

In particular, taking $\Delta_k(\gamma_k,x)=\p_k(\gamma_k,x)$ makes the last term
negative. Then changing $\Delta_2(b,x)$ to $-\p_2(b,x)$ makes only the third term
negative. Alternatively, changing $\Delta_1(a,x)$ to $-\p_1(a,x)$ makes only the
second term negative, or changing $\Delta_2(b,x)$ to $-\p_2(b,x)$ and $\Delta_1(A,x)$
to $-\p_1(A,x)$ makes only the first term negative.

To maximi{\s}e $\be$, we take three terms positive and one negative for each value
of $x$. The value of $\be$ will be maximi{\s}ed with the smallest term taken to be
negative. That is,
\begin{align}
\label{neat}
\be &= \sum_x p(x|\rho) \left\{\frac{\p_1(A,x)\p_2(B,x)}{\p(A,B,\rho)}
+\frac{\p_1(a,x)\p_2(B,x)}{\p(a,B,\rho)} \right. \nn & \quad
+\frac{\p_1(A,x)\p_2(b,x)}{\p(A,b,\rho)}
+\frac{\p_1(a,x)\p_2(b,x)}{\p(a,b,\rho)} \nn & \quad \left.
-2\min_{\alpha^x,\beta^x} \frac{\p_1(\alpha^x,x)\p_2(\beta^x,x)}{\p(\alpha^x,\beta^x,\rho)}\right\}.
\end{align}
Summing the first four terms in Eq.\ \eqref{neat} gives $\be$ equal to
\begin{equation}
\be_{\rm max} = 4-2\sum_x p(x|\rho) \min_{\alpha^x,\beta^x}
\frac{\p_1(\alpha^x,x)\p_2(\beta^x,x)}{\p(\alpha^x,\beta^x,\rho)}.
\end{equation}

To change the overall sign, so one term is positive and the remaining are negative, we can change the sign of both $\Delta_1(A,x)$ and $\Delta_1(a,x)$. Using this, we can ensure that the largest three terms in Eq.\ \eqref{sinz} positive, and the smallest is positive. Doing this for each value of $x$, we simply obtain the negative of what was obtained before, and so obtain $\be=-\be_{\rm max}$. We therefore find that the generali{\s}ation of the CHSH-Bell inequality for the case of postselection is $|\be|\le \be_{\rm max}$, as given in Eq.\ \eqref{geneq}. As the above argument is constructive, it shows how to choose probabilities
$p_k(\pm|\gamma_k,x)$ in order to saturate this inequality.
\end{proof}

Using this result, we can prove the necessary and sufficient condition for the CHSH-Bell inequality to hold for LHV theories with postselection. We first prove the necessary and sufficient condition in an alternative form.

\begin{lem}\label{lem2}
The condition that
\begin{equation}
\label{eq:con}
\frac{\p_1(\alpha,x)\p_2(\beta,x)}{\p(\alpha,\beta,\rho)}
\end{equation}
is independent of $\alpha$ and $\beta$ is necessary and
sufficient for $|\be|\le 2$ to be satisfied for all probabilities
$p_k(\pm|\gamma_k,x)$.
\end{lem}

\begin{proof}
In order for $\be$ to satisfy the usual CHSH-Bell inequality,
$|\be|\le 2$, the sum in Eq.\ \eqref{geneq} must be equal to 1. 
Given that the condition is satisfied, we may take $\alpha^x=\alpha$ and
$\beta^x=\beta$ for some arbitrary $\alpha$ and $\beta$ in Eq.\ \eqref{geneq},
giving
\begin{equation}
|\be| \le 4-2\sum_x p(x|\rho)
\frac{\p_1(\alpha,x)\p_2(\beta,x)}{\p(\alpha,\beta,\rho)}=2.
\end{equation}
Thus the condition is sufficient.

To prove necessity, we show that if the condition in Lemma \ref{lem2} is not satisfied, then
$|\be|\le 2$ may be violated. If the condition is not satisfied for $x=z$,
then select $\alpha$ and $\beta$ such that
\begin{equation}
\frac{\p_1(\alpha,z)\p_2(\beta,z)}{\p(\alpha,\beta,\rho)} >
\frac{\p_1(\alpha^{z},z)\p_2(\beta^{z},z)}{\p(\alpha^{z},\beta^{z},\rho)}.
\end{equation}
We than have
\begin{align}
&\sum_x p(x|\rho)\frac{\p_1(\alpha^x,x)\p_2(\beta^x,x)}{\p(\alpha^x,\beta^x,\rho)}\nn&\le
\sum_x p(x|\rho)\frac{\p_1(\alpha,x)\p_2(\beta,x)}{\p(\alpha,\beta,\rho)} \nn & \quad -
p(z|\rho)\left(\frac{\p_1(\alpha,z)\p_2(\beta,z)}{\p(\alpha,\beta,\rho)}
-\frac{\p_1(\alpha^{z},z)\p_2(\beta^{z},z)}{\p(\alpha^{z},\beta^{z},\rho)}
\right) \nn
&< 1.
\end{align}
As there exists a choice of probabilities $p_k(\pm|\gamma_k,x)$ which saturate the
inequality \eqref{geneq}, we can obtain $|\be|>2$. Hence we may obtain $|\be|>2$ if the condition in Lemma \ref{lem2} is violated, and it is therefore a necessary condition.
\end{proof}

We can now use Lemma \ref{lem2} to show that the factori{\s}ation condition \eqref{eq:req} is sufficient and necessary for the CHSH-Bell inequality. Specifically, we have the following theorem.

\begin{thm}
\label{the2}
The condition that the efficiency for each party factori{\s}es as
\begin{equation}
\label{eq:req2}
\p_{k}(\gamma_k,x) = \p_{k{\rm S}}(\gamma_k)\p_{k{\rm H}}(x),
\end{equation}
is necessary and sufficient for $|\be|\le 2$ to be satisfied for all probabilities
$p_k(\pm|\gamma_k,x)$ (provided no efficiency is zero).
\end{thm}

\begin{proof}
The proof proceeds simply by showing that the condition in Lemma \ref{lem2} is equivalent to \eqref{eq:req2}. First, it is trivial to show that Eq.\ \eqref{eq:req2} implies the condition. Using Eq.\ \eqref{eq:req2} gives
\begin{equation}
\frac{\p_1(\alpha,x)\p_2(\beta,x)}{\p(\alpha,\beta,\rho)}=\frac{\p_{1{\rm H}}(x)
\p_{2{\rm H}}(x)}{\p_{\rm R}(\rho)},
\end{equation}
which is independent of $\alpha$ and $\beta$.

To show that the condition of Lemma \ref{lem2} implies \eqref{eq:req2}, we simply need to use it to define the quantities $\p_{k{\rm S}}(\gamma_k)$ and $\p_{k{\rm H}}(x)$.
Let us define, for some $x_0$,
\begin{align}
\label{def1}
\p_{1{\rm S}}(\alpha) &:= {\p_1(\alpha,x_0)}/{\max_{\alpha'}\p_1(\alpha',x_0)}, \\
\p_{1{\rm H}}(x) &:= \max_\alpha\p_1(\alpha,x), \\
\p_{2{\rm S}}(\beta) &:= {\p_2(\beta,x_0)}/{\max_{\beta'}\p_1(\beta',x_0)}, \\
\label{def4}
\p_{2{\rm H}}(x) &:= \max_\beta\p_1(\beta,x).
\end{align}
It remains to show that these definitions satisfy \eqref{eq:req2}. We find that
\begin{align}
\p_{1{\rm S}}(\alpha)\p_{2{\rm S}}(\beta) &= \frac{\p_1(\alpha,x_0)\p_2(\beta,x_0)} {\max_{\alpha,\beta}\p_1(\alpha,x_0)\p_1(\beta,x_0)} \nn
&= \frac{\p(\alpha,\beta,\rho)}{\max_{\alpha,\beta}\p(\alpha,\beta,\rho)},
\end{align}
where in the second line we use the condition of Lemma \ref{lem2}, and $\rho$ may be arbitrary. Hence we find that $\p_{1{\rm S}}(\alpha)\p_{2{\rm S}}(\beta)$ is independent of $x_0$, and we therefore have
\begin{equation}
\label{eq:cru}
\p_{1{\rm S}}(\alpha)\p_{1{\rm H}}(x)\p_{2{\rm S}}(\beta)\p_{2{\rm H}}(x)
= \p_1(\alpha,x)\p_2(\beta,x).
\end{equation}

Rearranging Eq.\ \eqref{eq:cru} to isolate $\p_1(\alpha,x)$ gives
\begin{equation}
\p_1(\alpha,x) = \p_{1{\rm S}}(\alpha)\frac{\p_{1{\rm H}}(x)\p_{2{\rm S}}(\beta)\p_{2{\rm H}}(x)}{\p_2(\beta,x)}.
\end{equation}
Taking the maximum over $\alpha$ then yields
\begin{equation}
\p_{1{\rm H}}(x) = \frac{\p_{1{\rm H}}(x)\p_{2{\rm S}}(\beta)\p_{2{\rm H}}(x)}{\p_2(\beta,x)},
\end{equation}
where we have used $\max_\alpha\p_{1{\rm S}}(\alpha)=1$. Hence
\begin{equation}
\p_2(\beta,x) = \p_{2{\rm S}}(\beta)\p_{2{\rm H}}(x).
\end{equation}
Using this expression in Eq.\ \eqref{eq:cru} gives
\begin{equation}
\p_1(\alpha,x) = \p_{1{\rm S}}(\alpha)\p_{1{\rm H}}(x).
\end{equation}
Thus we have shown that the definitions \eqref{def1} to \eqref{def4} satisfy Eq.\ \eqref{eq:req2} provided the condition of Lemma \ref{lem2} holds.

Hence we have shown that the condition of Lemma \ref{lem2} and the condition \eqref{eq:req2} are equivalent, and therefore Theorem \ref{the2} follows from Lemma \ref{lem2}.
\end{proof}

\section{Quantum mechanics}
\label{sec:qm}
Next we consider the restriction on the loss for quantum mechanics, rather than
LHV theories. It might be thought that the case of quantum mechanics is equivalent, because a LHV theory can be thought of as an unentangled state, with the $x$ labeling orthogonal basis states. However, the case of quantum mechanics is slightly different, because we also need to consider all linear combinations of orthogonal basis states.

For quantum theory the probabilities of local measurement results are obtained via a positive operator-valued measure with elements $M_{k,\gamma_k}^{s_k}$, which corresponds to successful measurement result $s_k$ for measurement setting $\gamma_k$ for party $k$. We
also use the notation
\begin{equation}
M_{k,\gamma_k} := \sum_{s_k} M_{k,\gamma_k}^{s_k},
\end{equation}
for the operator corresponding to a successful measurement. Here the sum is over all successful measurement results $s_k$; the measurement operator for failure is $\openone-M_{k,\gamma_k}$.

The restriction on the loss in terms of hidden variables \eqref{eq:req} implies that
\begin{equation}
\label{eq:qure}
\p(\vec\gamma,\rho) = \p_{\rm R}(\rho)\prod_k \p_{k{\rm S}}(\gamma_k).
\end{equation}
This expression may be taken to be the definition of the restriction on the loss for quantum mechanics. This restriction may alternatively be expressed as
\begin{equation}
\label{conun}
\Tr\left[\left(\bigotimes_k M_{k,\gamma_k}\right) \rho\right] = \p_{\rm R}(\rho)\prod_k \p_{k{\rm S}}(\gamma_k).
\end{equation}
Let
\begin{equation}
\label{eq:mdef}
M := \frac{\bigotimes_k M_{k,\gamma_k}}{\prod_k \p_{k{\rm S}}(\gamma_k)}.
\end{equation}
The restriction \eqref{conun} implies that $\Tr(M\rho)=\p_{\rm R}(\rho)$ independent of $\vec\gamma$, and therefore $M$ is independent of $\vec\gamma$. This is why no subscript $\vec\gamma$ is given for $M$. Thus the restriction in Eq.\ \eqref{eq:qure} implies that the $M_{k,\gamma_k}/\p_{k{\rm S}}(\gamma_k)$ are independent of $\gamma_k$.

Using a similar method as for hidden variables, we can show that postselection with this form of loss cannot change the form of probability distributions obtained, either for the case of entangled or unentangled states. In particular, the result is as in the following theorem.

\begin{thm}
\label{the3}
For a Bell experiment on a quantum mechanical system, provided the loss is restricted by
\begin{equation}
p(\vec\gamma,\rho) = \p_{\rm R}(\rho)\prod_k \p_{k{\rm S}}(\gamma_k),
\end{equation}
the set of postselected probabilities it is possible to obtain without entanglement $\{\tilde p(\vec s|\vec \gamma,\rho)\}$ is identical to that which can be obtained without postselection or entanglement, and the set of postselected probabilities it is possible to obtain with entanglement is identical to that which can be obtained without postelection but with entanglement (provided no efficiency is zero).
\end{thm}

\begin{proof}
We define new measurement operators as
\begin{equation}
\tilde M_{k,\gamma_k}^\pm := (M_{k,\gamma_k})^{-1/2}M_{k,\gamma_k}^{s_k} (M_{k,\gamma_k})^{-1/2}.
\end{equation}
These measurement operators now give unit efficiency (so there is no postselection). We also define
\begin{equation}
\tilde \rho := \frac{M^{-1/2} \rho M^{-1/2}}{{\rm Tr}(M^{-1/2} \rho M^{-1/2})},
\end{equation}
where $M$ is defined as in Eq.\ \eqref{eq:mdef}. Because we restrict to the case of incomplete loss, the measurement operators are positive and can be inverted. The new measurement operators and states give exactly the same probabilities \emph{without} postselection as the original measurement operators and states did with postselection.

Furthermore, because $M$ is a local operator, it can not produce entangled states from unentangled states. Therefore the postselected Bell experiment with an unentangled state gives postselected probabilities identical to the non-postselected probabilities for a different Bell experiment with different measurement operators and a different unentangled state $\tilde \rho$.
\end{proof}

This result means that, for an unentangled state, all Bell inequalities must still be satisfied. This is exactly as we expect, because unentangled states may be regarded as equivalent to hidden variables. Similarly, inequalities that hold for entangled states (such as Tsirelson's bound) will be unaffected by the postselection.

For CHSH-Bell inequalities, it would be expected that Eq.\ \eqref{conun} is also a necessary condition due to the result for hidden-variable theories. However, it does not directly follow, and necessity needs to be proven separately. This is because one cannot arbitrarily choose the probabilities $p(\vec s|\vec\gamma,x)$ for given efficiencies. Therefore it is not necessarily possible to achieve the maximum value of Eq.\ \eqref{geneq}, as in the case of a general hidden-variable theory. Nevertheless, it is possible to show the result in the following theorem.

\begin{thm}
\label{the4}
For a quantum mechanical system, the condition that the efficiency for each party factori{\s}es as
\begin{equation}
\label{eq:req3}
\p(\vec\gamma,\rho) = \p_{\rm R}(\rho)\prod_k \p_{k{\rm S}}(\gamma_k),
\end{equation}
is necessary and sufficient for $|\be|\le 2$ to be satisfied for all probabilities
$p_k(\pm|\gamma_k,x)$ (provided no efficiency is zero).
\end{thm}

\begin{proof}
The sufficiency follows immediately from Theorem \ref{the3}, so it only remains to show necessity. The condition in the theorem is equivalent to Eq.\ \eqref{conun}, so we show necessity for Eq.\ \eqref{conun}. First consider the case where Eq.\ \eqref{conun} is violated. In that case, either $M_{\alpha}/\p_{1{\rm S}}(\alpha)$ is dependent on $\alpha$, or $M_{\beta}/\p_{2{\rm S}}(\beta)$ is dependent on $\beta$ (or both). We omit the subscripts $1$ and $2$ in the notation $M_{\alpha}$ and $M_{\beta}$ for simplicity. The party is simply indicated by the symbol used for the measurement setting ($\alpha$ for party 1 and $\beta$ for party 2).

If $M_{\alpha}/\p_{1{\rm S}}(\alpha)$ (for example) is not independent of $\alpha$, then there exist orthogonal states $\ket{\varphi_0}$ and $\ket{\varphi_1}$ such that
\begin{equation}
\label{eq:prop}
\bra{\varphi_0}M_{a}\ket{\varphi_0}\bra{\varphi_1}M_{A}\ket{\varphi_1} \ne 
\bra{\varphi_0}M_{A}\ket{\varphi_0}\bra{\varphi_1}M_{a}\ket{\varphi_1}.
\end{equation}
This result may be proven in the following way. Let us assume that all
orthogonal states give equality in \eqref{eq:prop}. Then, for any orthogonal
basis $\{\ket{\varphi_j}\}$, we have
\begin{equation}
\bra{\varphi_j}M_{a}\ket{\varphi_j}=\kappa\bra{\varphi_j}M_{A}\ket{\varphi_j},
\end{equation}
for some $\kappa$. Therefore $\Tr(M_{a})=\kappa\Tr(M_{A})$. Now we can take
$\{\ket{\varphi'_j}\}$ to be the basis which diagonali{\s}es $M_{a}-\kappa M_{A}$. Then
we have
\begin{equation}
\bra{\varphi'_j}M_{a}\ket{\varphi'_j}=\kappa'\bra{\varphi'_j}M_{A}\ket{\varphi'_j},
\end{equation}
so $\Tr(M_{a})=\kappa'\Tr(M_{A})$. Because $\Tr(M_{\alpha})$ is nonzero (these are
positive operators), we have $\kappa=\kappa'$. Therefore the diagonal elements
of $M_{a}-\kappa M_{A}$ must be zero in the basis $\{\ket{\varphi'_j}\}$. As this is
the basis which diagonali{\s}es $M_{a}-\kappa M_{A}$, we must have $M_{a}=\kappa M_{A}$.
Hence we find that equality in \eqref{eq:prop} implies $M_{a}=\kappa M_{A}$, so if
$M_{a}$ is not proportional to $M_{A}$, then there must exist orthogonal states
$\ket{\varphi_0}$ and $\ket{\varphi_1}$ such that \eqref{eq:prop} is satisfied.

We obtain the exact equivalent result if $M_{\beta}/\p_{2{\rm S}}(\beta)$ is not independent of $\beta$. Therefore, if Eq.\ \eqref{conun} is violated, we can select $\ket{\varphi_j}$
and $\ket{\chi_j}$ such that
\begin{align}
\bra{\varphi_0}M_{a}\ket{\varphi_0}\bra{\varphi_1}M_{A}\ket{\varphi_1} &\le 
\bra{\varphi_0}M_{A}\ket{\varphi_0}\bra{\varphi_1}M_{a}\ket{\varphi_1}, \nn
\bra{\chi_0}M_{b}\ket{\chi_0}\bra{\chi_1}M_{B}\ket{\chi_1} &\le 
\bra{\chi_0}M_{B}\ket{\chi_0}\bra{\chi_1}M_{b}\ket{\chi_1},
\end{align}
with strict inequality in at least one of these cases. We therefore have
\begin{align}
\label{eq:both}
&\bra{\varphi_0}M_{a}\ket{\varphi_0}\bra{\varphi_1}M_{A}\ket{\varphi_1}
\bra{\chi_0}M_{b}\ket{\chi_0}\bra{\chi_1}M_{B}\ket{\chi_1} \nn &<
\bra{\varphi_0}M_{A}\ket{\varphi_0}\bra{\varphi_1}M_{a}\ket{\varphi_1}
\bra{\chi_0}M_{B}\ket{\chi_0}\bra{\chi_1}M_{b}\ket{\chi_1}.
\end{align}

Now consider the density operator
\begin{equation}
\rho = \left( \ket{\varphi_0}\bra{\varphi_0} \otimes
\ket{\chi_0}\bra{\chi_0} +\ket{\varphi_1}\bra{\varphi_1} \otimes
\ket{\chi_1}\bra{\chi_1} \right)/2.
\end{equation}
For this density operator
\begin{align}
{\cal B} &= \frac 12 \sum_{j=0}^1 \left(\frac{\langle\Delta M_{A}\rangle_j
\langle\Delta M_{B}\rangle_j}{\Tr[(M_{A}\otimes M_{B})\rho]}+
\frac{\langle\Delta M_{A}\rangle_j\langle\Delta M_{b}\rangle_j}
{\Tr[(M_{A}\otimes M_{b})\rho]}\right. \nn & \left. \quad +
\frac{\langle\Delta M_{a}\rangle_j\langle\Delta M_{B}\rangle_j}
{\Tr[(M_{a}\otimes M_{B})\rho]}-
\frac{\langle\Delta M_{a}\rangle_j\langle\Delta M_{b}\rangle_j}
{\Tr[(M_{a}\otimes M_{b})\rho]}\right),
\end{align}
where $\Delta M_{k,\gamma_k}=M_{k,\gamma_k}^+-M_{k,\gamma_k}^-$, and
$\langle \cdots \rangle_j$ indicates the expectation value using the
state $\ket{\varphi_j}$ (for Alice) or $\ket{\chi_j}$ (for Bob).
Given the $M_{k,\gamma_k}$, consider the measurement operators for the individual
results given by
\begin{align}
M_{A}^+ &= M_{A}, \hspace{1.57cm} M_{a}^+ = \ket{\varphi_0}\bra{\varphi_0}\langle M_{a}\rangle_0, \nn
M_{A}^- &= 0, \hspace{1.97cm}
M_{a}^- = M_{a}-\ket{\varphi_0}\bra{\varphi_0}\langle M_{a}\rangle_0, \nn
M_{B}^+ &= \ket{\varphi_0}\bra{\varphi_0}\langle M_{B}\rangle_0, \hspace{2.11cm}
M_{b}^+ = M_{B}, \nn
M_{B}^- &= M_{B}-\ket{\varphi_0}\bra{\varphi_0}\langle M_{B}\rangle_0, \hspace{1.1cm} M_{b}^- = 0.
\end{align}
Using these measurement operators gives
\begin{align}
\langle \Delta M_{a} \rangle_j &= (-1)^j\langle M_{a} \rangle_j, \quad
\langle \Delta M_{A} \rangle_j = \langle M_{A} \rangle_j, \nn
\langle \Delta M_{b} \rangle_j &= \langle M_{b} \rangle_j, \quad
\langle \Delta M_{B} \rangle_j = (-1)^j\langle M_{B} \rangle_j.
\end{align}
This gives the postselected Bell quantity as
\begin{align}
{\cal B} &= 4-\frac{2\langle M_{a}\rangle_0\langle M_{b}\rangle_0}
{\langle M_{a}\rangle_0\langle M_{b}\rangle_0+
\langle M_{a}\rangle_1\langle M_{b}\rangle_1} \nn & \quad
-\frac{2\langle M_{A}\rangle_1\langle M_{B}\rangle_1}
{\langle M_{A}\rangle_0\langle M_{B}\rangle_0+
\langle M_{A}\rangle_1\langle M_{B}\rangle_1}.
\end{align}
Equation \eqref{eq:both} may be rearranged to give
\begin{align}
&\frac{\langle M_{a}\rangle_0\langle M_{b}\rangle_0}
{\langle M_{a}\rangle_0\langle M_{b}\rangle_0+
\langle M_{a}\rangle_1\langle M_{b}\rangle_1} \nn &
+\frac{\langle M_{A}\rangle_1\langle M_{B}\rangle_1}
{\langle M_{A}\rangle_0\langle M_{B}\rangle_0+
\langle M_{A}\rangle_1\langle M_{B}\rangle_1}<1.
\end{align}
We therefore obtain $\be>2$.

Thus we find that, if Eq.\ \eqref{eq:qure} is not satisfied, for given $M_{k,\gamma_k}$ there exists a separable state and a set of measurement operators $M_{k,\gamma_k}^\pm$ such that $\be>2$. Eq.\ \eqref{eq:qure} is therefore a necessary condition for the CHSH-Bell inequality to be satisfied for separable states. As Eq.\ \eqref{eq:qure} is equivalent to the condition in the theorem, we have proven both necessity and sufficiency as required.
\end{proof}

For entangled states, $\be$ is limited by Tsirelson's bound of $2\sqrt{2}$ \cite{tsirelson}. As explained above, the restriction given by Eq.\ \eqref{eq:qure} is sufficient for Tsirelson's bound to be satisfied due to the result in Theorem \ref{the3}. However, it turns out that it is not necessary. In particular, we find that Tsirelson's bound is not violated if there is no loss on one side, and moderate loss on the other side.  Let $M_{a}=M_{A}=M_{b}=\openone$, and let
\begin{equation}
M_{B}=
\left[ {\begin{array}{*{20}c}
   1 & 0  \\
   0 & p  \\
\end{array}} \right].
\end{equation}
We have performed numerical maximi{\s}ations over the measurements and states for
a range of values of $p$, and the results are shown in Fig.\ \ref{counter}.
For small values of $p$, below about $0.17$, there is violation of Tsirelson's
bound, but for larger values no violation of Tsirelson's bound is achieved.
This indicates that there is not a simple necessary and sufficient condition in
the case of Tsirelson's bound. Whether Tsirelson's bound can be violated depends
on the particular value of the loss.

\begin{figure}
\centering
\includegraphics[width=0.45\textwidth]{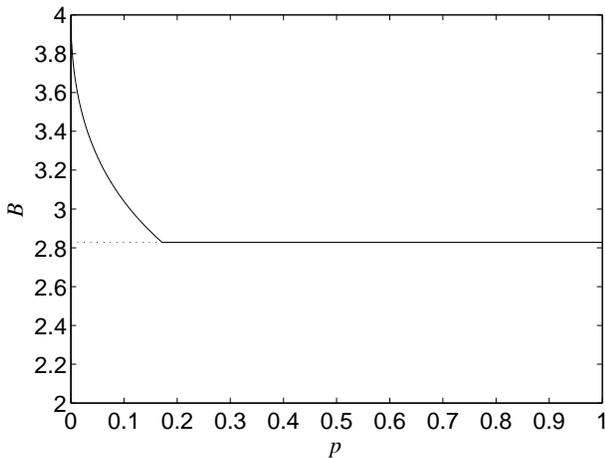}
\caption{The numerically found maximal values of $\be$ with $M_a=M_A=M_b=\openone$ and $M_B={\rm diag}(1,p)$. Tsirelson's bound is shown as the dotted line for comparison.}
\label{counter}
\end{figure}

An interesting fact is that if the condition \eqref{eq:qure} is violated for both parties, in the sense that $M_{k,\gamma_k}/\p_{k{\rm S}}(\gamma_k)$ is dependent on $\gamma_k$ for both Alice and Bob, then it is always possible to find measurements such that Tsirelson's bound is violated. A scheme for doing this is presented in Appendix \ref{symtsi}.

\section{Postselection that violates Tsirelson's bound}
\label{sec:tsi}
A central motivation for this work was to determine forms of loss that can violate Tsirelson's bound while ensuring that the CHSH-Bell inequality is still valid, in order to generali{\s}e Ref.\ \cite{cab}. The result of the previous section is that the form of postselection that is necessary for the CHSH-Bell inequality to hold for unentangled states (as shown in Theorem \ref{the4}) also implies that Tsirelson's bound holds for entangled states (as follows from Theorem \ref{the3}).

This raises the question of how the postselection in Ref.\ \cite{cab} differs from the postselection used here. The difference is that the postselection relies on the experimenter knowing what the state is. This method for constructing Bell inequalities is critici{\s}ed in Ref.\ \cite{anne}, because it introduces an additional assumption beyond locality and realism. If an entangled state is incorrectly assumed, but the actual state is unentangled, then the method of Ref.\ \cite{cab} still yields a violation of the Bell inequality (See Appendix \ref{sec:cab}). We therefore avoid the approach of assuming a state here, and simply allow postselection that may depend on the state and measurement setting. 

A crucial subtlety in our results is that the condition \eqref{eq:qure} is necessary for the CHSH-Bell inequality to be satisfied for unentangled states, provided we consider arbitrary measurement schemes for the given form of loss. This does not eliminate the
possibility that there are particular measurement schemes with postselection
such that the Bell inequality is satisfied for unentangled states but
Tsirelson's bound may be violated for entangled states. We consider such a scheme in this section.

We emphasi{\s}e that this scheme does not violate the Bell inequality with unentangled states \emph{provided} the measurements are acting as expected. But, if there is an underlying hidden variable theory, or equivalently if the measurements are not acting as expected on the underlying quantum state, then the CHSH-Bell inequality can be violated without entanglement.

The advantage of this postselection is that it also increases the Bell quantity that may be obtained with entanglement. The standard result that the efficiency must be at least $2(\sqrt 2 - 1)\approx 82.8\%$ with maximally entangled states \cite{garg} is based upon the assumption that Tsirelson's bound still holds for entangled states. If the Bell quantity is also enhanced for entangled states, then the efficiency may be lower before the value that is possible with a LHV theory reaches that possible with entangled states.

The example we consider is where Alice and Bob share the two-qubit entangled state
\begin{equation}
|\psi_-\rangle_{12}={\cal N}(|u\rangle_1|u\rangle_2-|v\rangle_1|v\rangle_2)
\label{eq:es}
\end{equation}
where ${\cal N}$ is the normali{\s}ation factor. The states $\ket u$ and $\ket v$
are assumed to be non-orthogonal, with the real inner product $\kappa$. We define
a local transformation $R(\theta)$ which acts as $R(\theta)|u\rangle=
|u\rangle$ and $R(\theta)|v\rangle=e^{i \theta}|v\rangle$.

The transformation $R(\theta)$ is non-unitary and cannot be
reali{\s}ed deterministically with a nonzero $\kappa$. Alice performs the local
operation $R(\theta)$, and Bob applies $R(\phi)$. The different measurement settings for Alice and Bob are achieved by using
different rotation angles $\theta$ and $\phi$. This
non-deterministically transforms the state $|\psi_-\rangle_{12}$ to
\begin{equation}
|\psi_-^\prime\rangle_{12}= {\cal N}^\prime(|u\rangle_1|u\rangle_2-e^{i(\theta+\phi)}
|v\rangle_1|v\rangle_2).
\end{equation}
Alice and Bob then perform orthogonal measurements on their respective qubits
using the basis $|\pm\rangle={\cal N}_\pm(|u\rangle\pm|v\rangle)$, and there are
four possible cases of the combined measurement results. The probabilities for these four cases are
\begin{align}
P_{++}(\theta,\phi) &= \frac{(1+\kappa)^2(1-\cos[\theta+\phi])}
{4(1-\kappa^2\cos[\theta+\phi])},\nonumber\\
P_{+-}(\theta,\phi) &= P_{-+}(\theta,\phi)=\frac{(1-\kappa^2)
(1+\cos[\theta+\phi])}
{4(1-\kappa^2\cos[\theta+\phi])},\nonumber\\
P_{--}(\theta,\phi) &= \frac{(1-\kappa)^2(1-\cos[\theta+\phi])}
{4(1-\kappa^2\cos[\theta+\phi])},
\label{eq:1234}
\end{align}
where $P_{s_1s_2}(\theta,\phi)= |\langle \psi_-^\prime
|s_1\rangle_1|s_2\rangle_2|^2$.

 The Bell function $\be$ is then
obtained as
\begin{equation}
\be=C(\theta_A,\phi_B)+C(\theta_A,\phi_b)+C(\theta_a,\phi_B)-C(\theta_a,\phi_b),
\end{equation}
where
\begin{align}
C(\theta,\phi) &=P_{++}-P_{+-}-P_{-+}+P_{--} \nn
&=\frac{\kappa^2-\cos(\theta+\phi)}{1-\kappa^2\cos(\theta+\phi)}.
\end{align}
The maximum is obtained for
\begin{equation}
(\theta_A+\theta_a+\phi_B+\phi_b)=\theta_a-\theta_A=\phi_b-\phi_B=\Theta,
\end{equation}
for some parameter $\Theta$. The Bell function then depends only on $\Theta$:
\begin{equation}
\be=3C\left(\Theta/2,0\right)-C\left(3\Theta/2,0\right).
\end{equation}
There does not appear to
be an analytic solution for $\Theta$, but a good approximation is given by
$\Theta = \pi[17+\cos(\pi\kappa)]/12$.

To determine the loss that is required for the operation $R(\theta)$ on $|u\rangle$ and $|v\rangle$, consider the singular value
decomposition \cite{SVD} of the matrix representing $R(\theta)$, i.e.,
$U_+ W U_-^\dagger$ (omitting a global phase) where
\begin{align}
&U_\pm=\frac 1{\sqrt{b^2+(r\pm d)^2}}\left[\begin{array}{*{10}cc}-ib & r\pm d \\
-(r\pm d) & ib \\ \end{array}\right], \label{eq:UU} \\
&W=\left[ \begin{array}{*{10}cc} \sqrt{\frac{a-r}2} & 0,\\
0 & \sqrt{\frac{a+r}2} \label{eq:V} \\ \end{array}\right],
\end{align}
with
$b=\sin\theta\big( \sqrt{(1+\kappa)/(1-\kappa)}
-\sqrt{(1-\kappa)/(1+\kappa)} \big)$,
$d=4\kappa\sin^2(\theta/2)/(1-\kappa^2)$,
$a=2+\kappa d$, and $r=\sqrt{a^2-4}$. This therefore gives the minimum detector efficiency as
\begin{equation}
\eta = \sqrt{\frac{a-r}{a+r}}.
\end{equation}
The maximum value of $\be$ possible with LHV theories is then $4/\eta-2$ \cite{garg}.

The values of the quantity $\be$ for the initial state $|\psi_-\rangle_{12}$, as
well as the maximum values for separable states, and the maximum values for LHV theories, are plotted in Fig.\ \ref{f2}. The maximum value for separable states was obtained by numerical optimi{\s}ation, and the maximum value for LHV theories was obtained using the formula $4/\eta-2$.
For $|\psi_-\rangle_{12}$, the value for entangled states varies from
Tsirelson's bound for $\kappa=0$ to a maximum of 4 in the limit $\kappa\to 1$.
On the other hand, $\be$ initially has a maximum of $\sqrt 2$ for separable
states. It increases with $\kappa$, but still does not exceed 2 for
$\kappa\lesssim 0.357$. For this value of $\kappa$, $\be$ is almost 3 (2.966)
for $|\psi_-\rangle_{12}$. On the other hand, the maximum value of $\be$ for LHV theories increases from 2 for $\kappa=0$ (no loss), and crosses the line for entangled states at $\kappa\approx 0.124$. At this value of $\kappa$, $\eta\approx 82.6\%$, slightly below the limit of $82.8\%$ derived in Ref.\ \cite{garg}.

\begin{figure}
\begin{centering}
{\includegraphics[width=0.45\textwidth]{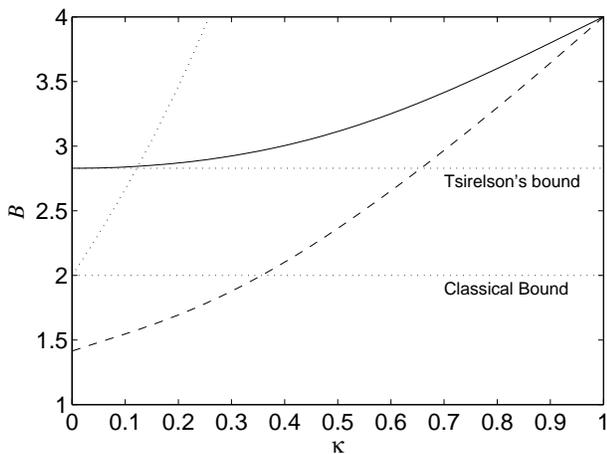}} \end{centering} \caption{The CHSH-Bell
quantity $\be$ using postselection. The value for $|\psi_-\rangle_{12}$ is shown
as the solid line, the maximum value for separable states is shown as the
dashed line, and the maximum value for LHV theories is shown as the dotted curve.} \label{f2}
\end{figure}

We have also considered more general measurements that do not have the simple interpretation given above. These measurements are numerically optimi{\s}ed to maximi{\s}e the value of $\be$ for entangled states with a given level of loss. The results for this scheme are shown in Fig.\ \ref{f3}. For this numerically optimi{\s}ed scheme the value of $\be$ can be greater than 3 (3.0046) for entangled states before it reaches 2 for separable states. The value of $\eta$ when the value of $\be$ for LHV theories reaches that for entangled states is slightly lower than before, at $82.4\%$.

\begin{figure}
\begin{centering}
{\includegraphics[width=0.45\textwidth]{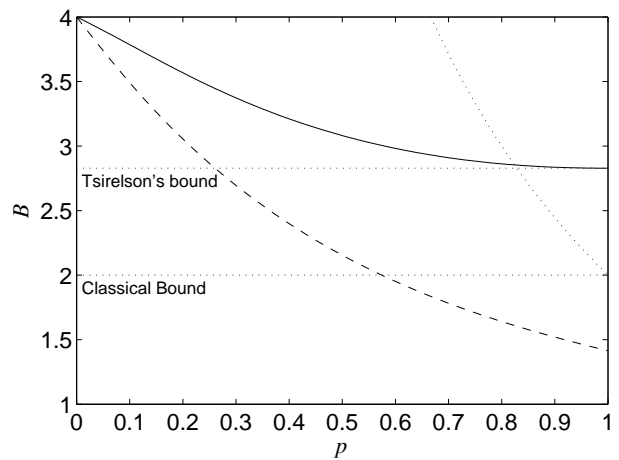}} \end{centering} \caption{The CHSH-Bell
quantity $\be$ using postselection with a numerically optimi{\s}ed scheme. The value for the maximally entangled state is shown
as the solid line, the maximum value for separable states is shown as the
dashed line, and the maximum value for LHV theories is shown as the dotted curve.} \label{f3}
\end{figure}

Methods for general non-unitary transformations such as $R(\theta)$ have
been presented in Ref.\ \cite{tera05}, and methods for performing single-qubit
non-unitary transformations in optical systems in Ref.~\cite{bergou00}. 
An experiment for our proposal can be performed using current technology with
the photon-polari{\s}ation qubit basis, namely, $|H\rangle$ and $|V\rangle$ where
$|H\rangle$ is the horizontal-polari{\s}ation state of a photon and $|V\rangle$
vertical. The two non-orthogonal states $|u\rangle$ and $|v\rangle$ can be
defined as $|u\rangle=\lambda|H\rangle+\sqrt{1-\lambda^2}|V\rangle$ and
$|v\rangle=\lambda|H\rangle-\sqrt{1-\lambda^2}|V\rangle$, where $\lambda$ is
assumed to be a real value larger than $1/\sqrt{2}$. In this case, the
orthogonal basis states, $|+\rangle=|H\rangle$ and $|-\rangle=|V\rangle$,
are obtained and $\kappa=\langle u|v\rangle=2\lambda^2-1$. The entangled
state $|\psi_-\rangle_{12}$ turns out to be equal to ${\cal N}(|u\rangle_1|
u\rangle_2-|v\rangle_1|v\rangle_2)=(|H\rangle_1|V\rangle_2+|V\rangle_1|H
\rangle_2)/\sqrt{2}$. It is well known that such a two-qubit entangled state can
be generated using spontaneous parametric down-conversion \cite{Kwiat95}.

Via a singular value decomposition, $R(\theta)$ may be expressed as
$U_+ W U_-^\dagger$, where the $U_\pm$ are unitary and $W$ is diagonal, as given in Eqs.\ \eqref{eq:UU} and \eqref{eq:V}, respectively.
The unitary operations correspond to single qubit rotations, and can be
implemented using linear optical elements \cite{KLM}. After $U_-^\dagger$ is
performed on the corresponding qubit, the non-unitary component $W$ can be
achieved using selective absorption at one polari{\s}ation. Provided the photon is
detected in the final measurement, $W$ will have been successfully applied to
the qubit. Then $U_+$ is applied as discussed above. The final measurement is
then just a measurement in the polari{\s}ation basis, which can be achieved using
a polari{\s}ing beam splitter and photodetectors.

\section{Conclusions}
\label{sec:conc}
Our results provide clear guidelines for determining if the sampling is fair in any particular Bell experiment. If the goal is simply to demonstrate entanglement, rather than disprove LHV theories, then the condition required is given by Eq.\ \eqref{eq:qure}, which can be experimentally tested. If it can be demonstrated to hold, then any violation of Bell inequalities is due to entanglement. This is relevant to quantum key distribution, where security may be proven by violation of Bell inequalities \cite{ekert,shor,acin,barrett,pironio}. In cases where there is loss, but the sampling is tested and found to be fair, then the key distribution should still be secure.

In the case where hidden variables are allowed, then the relevant condition is that the efficiency factori{\s}es as in Eq.\ \eqref{eq:req}. Any condition that depends on the hidden variables can not be proven to hold, because it is possible that it might be violated for values of the hidden variable that it is not possible to prepare. However, it is possible to falsify it. The great advantage of providing a \emph{necessary} condition, as we have done, is that if it can be shown not to hold, then the sampling is shown to be of a form that invalidates the CHSH-Bell inequality. In contrast, if the condition that is tested is \emph{not} necessary, then testing it is not useful. Showing that it does not hold does not show that the sampling is of a form that invalidates the CHSH-Bell inequality, and it cannot be conclusively shown to hold. Thus our results put testing of the sampling in Bell experiments \cite{aden} on a rigorous basis.

These results also cast light on the effort to develop ways of enhancing Bell violations using postselection, as discussed in Refs.\ \cite{pawlowski,cab}. Here we have found that the condition on the loss that is necessary to prevent violation of the CHSH-Bell inequality with postselection is also sufficient to ensure that Tsirelson's bound is not violated. Our results show that this form of sampling cannot be used to enhance violation of Bell inequalities beyond what is possible without postselection.

This suggests that is should not be possible to violate Tsirelson's bound with postselection while retaining a valid Bell inequality (as claimed in \cite{cab}). On the other hand, it is possible for postselection to enhance the violation of Bell inequalities up to Tsirelson's bound without invalidating them. For example, Procrustean entanglement concentration \cite{Bennett} has been demonstrated to enhance the violation of Bell inequalitites \cite{Kwiat01}. Because Procrustean entanglement concentration corresponds to fair sampling, this postselection does not invalidate the Bell inequalities for unentangled states, and also does not lead to violation of Tsirelson's bound for entangled states.

Our results do not entirely rule out the possibility of violating Tsirelson's bound while the CHSH-Bell inequality still holds for unentangled states. However, this requires examining the measurement used, rather than simply relying on the form of loss. We have found schemes that allow the Bell quantity to be as large as 3 before the CHSH-Bell inequality is violated with separable states. These schemes also allow the efficiency to be below the limit of $82.8\%$ derived by Ref.\ \cite{garg} before the value of $\be$ possible with LHV theories reaches that obtained for entangled states. This is because it is assumed in Ref.\ \cite{garg} that the loss cannot enhance violation of the Bell inequality beyond Tsirelson's bound.

\acknowledgments
We thank R. Cleve for useful discussions. This work was supported by Australian Research Council and Queensland State Government, the WCU program, and the KOSEF grant funded by the Korean government (MEST) (R11-2008-095-01000-0).

\appendix

\section{Efficiency independent of the measurement setting for one state}
\label{sec:example}
As an example of a case where $|{\cal B}|> 2$ when the efficiency is independent of the measurement settings for the state under consideration, but Eq.\ \eqref{eq:qure} is violated, consider the state
\begin{equation}
\rho = \frac 12 \left( \ket 0 \bra 0 + \ket 1 \bra 1 \right).
\end{equation}
For Bob we have the measurement operators
\begin{align}
M_B^- =
\left[ \begin{array}{*{20}c}
   0 & 0 \\
   0 & 0 \\
\end{array} \right], &\qquad
M_B^+ =
\left[ \begin{array}{*{20}c}
   1 & 0 \\
   0 & 1 \\
\end{array} \right], \\
M_b^- =
\left[ \begin{array}{*{20}c}
   0 & 0 \\
   0 & 1 \\
\end{array} \right], &\qquad
M_b^+ =
\left[ \begin{array}{*{20}c}
   1 & 0 \\
   0 & 0 \\
\end{array} \right].
\end{align}
That is, Bob either always returns the result $+$, or performs a measurement in the computational basis, and returns $+$ or $-$ for $0$ or $1$, respectively. Alice's measurement operators are
\begin{align}
M_A^- =
\left[ \begin{array}{*{20}c}
   0 & 0 \\
   0 & 0 \\
\end{array} \right], &\qquad
M_A^+ =
\left[ \begin{array}{*{20}c}
   1 & 0 \\
   0 & 0 \\
\end{array} \right], \\
M_a^- =
\left[ \begin{array}{*{20}c}
   0 & 0 \\
   0 & 0 \\
\end{array} \right], &\qquad
M_a^+ =
\left[ \begin{array}{*{20}c}
   0 & 0 \\
   0 & 1 \\
\end{array} \right].
\end{align}
That is, for measurement $A$ Alice returns the result $+$ for $\ket 0$, or there is failure for $\ket 1$. For measurement $a$ Alice returns $+$ for $\ket 1$, or failure for $\ket 0$. Experimentally this could correspond to a beam splitter sending two different polari{\s}ations to different photodetectors, one of which always fails.

For this example, the probability of success is always $1/2$, and we obtain ${\cal B}=4$. Thus we obtain the maximum possible violation of the Bell-CHSH inequality and Tsirelson's bound with a separable state, and with equal probabilities of success for each measurement setting. This demonstrates that equal probabilities of success for each measurement setting for one state does not imply that the loss preserves the CHSH-Bell inequality. It is also important to note that unequal efficiencies for detectors is sufficient to invalidate the Bell-CHSH inequality and Tsirelson bound. This is an important case, because one would not expect detectors to be identical experimentally.

We may also use this case to provide a counterexample which demonstrates that the condition \eqref{eq:qure} does not imply \eqref{eq:req}. Let the hidden variable take the values 0 and 1, and give the probabilities of the measurement results and successes as above. However, take the quantum system to be 1-dimensional, so there are no measurable features. The probability of failure for all combinations of measurement settings is equal to $1/2$, regardless of the state, so it is clear that \eqref{eq:qure} is satisfied. However, the hidden variables do not satisfy \eqref{eq:req}, and the Bell inequality is violated. In fact, for this hidden variable theory, Tsirelson's bound is violated as well, despite \eqref{eq:qure} being satisfied. This is because the proof in Sec.\ \ref{sec:qm} assumes that probabilities are obtained using quantum measurement theory, not a hidden variable theory.

\section{Violation of Tsirelson's bound}
\label{symtsi}
Here we show that Tsirelson's bound may be violated when $M_{k,\gamma_k}/\p_{k{\rm S}}(\gamma_k)$ is dependent on $\gamma_k$ for both Alice and Bob. Note that this is a slightly stronger condition than violation of \eqref{eq:qure}, because \eqref{eq:qure} can be violated with dependence on just one side.

Define the unitary operators $U_A$ and $U_B$ to be those that diagonalise
$(M_a)^{-1/2} M_A (M_a)^{-1/2}$ and $(M_b)^{-1/2} M_B (M_b)^{-1/2}$,
respectively. We define the new operators
\begin{align}
\tilde M_A &\equiv \frac{U_A(M_a)^{-1/2} M_A (M_a)^{-1/2}U_A\dg}{\max [{\rm eig}(M_a)^{-1/2} M_A (M_a)^{-1/2} ]}, \nn
\tilde M_A^\pm &\equiv \frac{U_A(M_a)^{-1/2} M_A^\pm (M_a)^{-1/2}U_A\dg}{\max [{\rm eig}(M_a)^{-1/2} M_A (M_a)^{-1/2} ]}, \nn
\tilde M_B &\equiv \frac{U_B(M_b)^{-1/2} M_B (M_b)^{-1/2}U_B\dg}{\max [{\rm eig}(M_b)^{-1/2} M_B (M_b)^{-1/2} ]}, \nn
\tilde M_B^\pm &\equiv \frac{U_B(M_b)^{-1/2} M_B^\pm (M_b)^{-1/2}U_B\dg}{\max [{\rm eig}(M_b)^{-1/2} M_B (M_b)^{-1/2} ]}, \nonumber
\end{align}
\begin{align}
\tilde M_a &\equiv U_A(M_a)^{-1/2} M_a (M_a)^{-1/2}U_A\dg = \openone, \nn
\tilde M_a^\pm &\equiv U_A(M_a)^{-1/2} M_a^\pm (M_a)^{-1/2}U_A\dg, \nn
\tilde M_b &\equiv U_B(M_b)^{-1/2} M_b (M_b)^{-1/2}U_B\dg = \openone, \nn
\tilde M_b^\pm &\equiv U_B(M_b)^{-1/2} M_b^\pm (M_b)^{-1/2}U_B\dg.
\end{align}
In addition, the system state $\rho$ is replaced with
\begin{equation}
\tilde \rho \equiv \frac{(U_A\otimes U_B)(M_a \otimes M_b)^{1/2}
\rho(M_a \otimes M_b)^{1/2}(U_A\dg \otimes U_B\dg)}
{\Tr[(M_a \otimes M_b) \rho]}.
\end{equation}
The measurements in terms of these new operators on the new state $\tilde\rho$
must give the same normali{\s}ed $C(\alpha,\beta)$ as the original operators, and
therefore the same value for the CHSH-Bell quantity. In addition, the operators
$\tilde M_A$, $\tilde M_a$, $\tilde M_B$, and $\tilde M_b$ are all diagonal
(and $\tilde M_a=\openone$ and $\tilde M_b=\openone$).

The normali{\s}ations given for $M_A$, $M_A^\pm$, $M_B$ and $M_B^\pm$ ensure that
these operators do not have eigenvalues larger than 1. The unitaries $U_A$ and
$U_B$ can also be taken to be those that sort the eigenvalues in decreasing
order. Now we restrict to a two-dimensional subspace where $M_A$ and $M_B$ are
of the form
\begin{align}
M_A &= \left( \begin{array}{*{20}c}
   {1} & {0}  \\
   {0} & {p_A}  \\
\end{array} \right), \\
M_B &= \left( \begin{array}{*{20}c}
   {1} & {0}  \\
   {0} & {p_B}  \\
\end{array} \right).
\end{align}
Consider taking the state $\tilde\rho$ to be a maximally entangled state
\begin{align}
\tilde\rho &= \ket{\psi}\bra{\psi}, \nn
\ket\psi &= -i(\ket{00}+\ket{11})/2+(\ket{01}+\ket{10})/2.
\end{align}
We can also take the operators $M_\alpha^{\pm}$ such that
\begin{align}
\Delta \tilde M_A &= -M_A^{1/2} Z M_A^{1/2}, \nn
\Delta \tilde M_a &= Y, \nn
\Delta \tilde M_B &= \frac 1{\sqrt 2} M_B^{1/2} (Y-Z) M_B^{1/2}, \nn
\Delta \tilde M_b &= \frac 1{\sqrt 2}(Y+Z).
\end{align}
Here $Y$ and $Z$ are the usual Pauli operators $\sigma_y$ and $\sigma_z$.
Using these operators we find that
\begin{align}
&\frac{\Tr(\Delta \tilde M_A\Delta \tilde M_B\tilde \rho)}{\Tr(\tilde M_A\tilde M_B\tilde \rho)} \nn & \quad = \frac {1-p_A+2(1+p_A)\sqrt{p_B}+(p_A-1)p_B}{\sqrt 2(1+p_A)(1+p_B)} \nn
&\frac{\Tr(\Delta \tilde M_A\Delta \tilde M_b\tilde \rho)}{\Tr(\tilde M_A\tilde M_b\tilde \rho)} = \frac 1{\sqrt 2} \nn
&\frac{\Tr(\Delta \tilde M_a\Delta \tilde M_B\tilde \rho)}{\Tr(\tilde M_a\tilde M_B\tilde \rho)} = \frac 1{\sqrt 2} \nn
&\frac{\Tr(\Delta \tilde M_a\Delta \tilde M_b\tilde \rho)}{\Tr(\tilde M_a\tilde M_b\tilde \rho)} = -\frac 1{\sqrt 2}.
\end{align}
The Bell quantity is then
\begin{align}
{\cal B} &= 2\sqrt 2 + \frac{\sqrt 2 (1-\sqrt{p_B})(\sqrt{p_B}-p_A)}{(1+p_A)(1+p_B)}.
\end{align}
This is greater than Tsirelson's bound provided that $p_A<\sqrt{p_B}<1$.

In the case that $p_B<\sqrt{p_A}<1$, it is possible to obtain violation of
Tsirelson's bound by exchanging the roles of Alice and Bob. That is,
\begin{align}
\Delta \tilde M_A &= \frac 1{\sqrt 2} M_A^{1/2} (Y-Z) M_A^{1/2}, \nn
\Delta \tilde M_a &= \frac 1{\sqrt 2}(Y+Z), \nn
\Delta \tilde M_B &= -M_B^{1/2} Z M_B^{1/2}, \nn
\Delta \tilde M_b &= Y,
\end{align}
so
\begin{align}
{\cal B} &= 2\sqrt 2 + \frac{\sqrt 2 (1-\sqrt{p_A})(\sqrt{p_A}-p_B)}{(1+p_A)(1+p_B)}.
\end{align}

Hence, provided $p_A\ne 1$ and $p_B\ne 1$ (so $M_{k,\gamma_k}/\p_{k{\rm S}}(\gamma_k)$ is dependent on $\gamma_k$ for both Alice and Bob), it is possible to violate Tsirelson's
bound. If one of $p_A$ or $p_B$ is equal to 1, then this approach does not yield
violation of Tsirelson's bound.

\section{Postselection of Ref.\ \cite{cab}}
\label{sec:cab}
Here we show that the approach to postselection in Ref.\ \cite{cab} allows
violation of the CHSH-Bell inequality with unentangled states, if an entangled state is incorrectly assumed. Ref.\ \cite{cab}
considers a three-qubit entangled state, where measurements are made in either
the $X$ or $Z$ basis. Ref.\ \cite{cab} considers the Greenberger-Horne-Zeilinger
(GHZ) state $\ket{\Psi}=(\ket{+++}_y+\ket{---}_y)/\sqrt 2$,
where $\ket{\pm}_y$ are $Y$ eigenstates. The qubits are labelled $i$, $j$ and $k$
according to the following prescription. When $Z$ is measured on each qubit,
either two results are $-1$ and one is $+1$, or all three are $+1$. If two
results are $-1$, the corresponding qubits are taken to be $i$ and $j$. If all
three results are $+1$, the labels $i$, $j$ and $k$ are assigned randomly.

In practice, $Z$ is not always measured on all three qubits, so it is not
possible to apply this method directly. Instead, one could force qubit 3 to be
$k$ by measuring $Z$ on it and postselecting on a measurement result of 1 (as
proposed by Ref.\ \cite{cab}). Alternatively, one could randomly select a qubit, and
force it to be $k$ by this method.

The form of the CHSH-Bell inequality in Ref.\ \cite{cab} is $|\be|\leq2$, with
\begin{equation}
\be_x = C(Z_i,Z_j)-x_kC(Z_i,X_j)-x_kC(X_i,Z_j)-C(X_i,X_j).
\end{equation}
Here subscripts are used to indicate the operator acting on the corresponding
qubit, and $x_k$ is the result which would be obtained by measuring $X_k$. As
$X_k$ is not measured, the value of $x_k$ is inferred from the results
of the measurements on qubits $i$ and $j$. One takes $x_k$ to be
the negative of the product of the measurement results for qubits $i$ and $j$
(if $X$ is measured on one and $Z$ on the other).
Because of this method for selecting $x_k$,
\begin{align}
\label{perfect}
-x_kC(Z_i,X_j)=-x_kC(X_i,Z_j)=1.
\end{align}
For the GHZ state, one also obtains $C(Z_i,Z_j)=1$ and $C(X_i,X_j)=-1$, so
$\be_x=4$.

However, if one uses this method for selecting $i$, $j$ and $k$, as well as for
selecting the value of $x_k$, with an unentangled state, then the inequality
$|\be_x|\le2$ can still be violated. Consider, for example, the state
\begin{align}
\rho &= \frac 14(\ket{--+}_z\bra{--+}+\ket{-+-}_z\bra{-+-}\nn & \quad
+\ket{+--}_z\bra{+--}+\ket{+++}_z\bra{+++}).
\end{align}
It is found that, regardless of which qubit is taken to be $k$, the $Z_i$ and
$Z_j$ measurements are perfectly correlated, so $C(Z_i,Z_j)=1$. There are no
correlations in the $X$ measurements, so $C(X_i,X_j)=0$, and
Eq.\ \eqref{perfect} still holds (because that is due to the method of
determining $x_k$). We therefore find $\be_x=3$, which violates both the Bell
inequality and Tsirelson's bound.


\begin{thebibliography}{99}
\bibitem{epr} A. Einstein, B. Podolsky, and N. Rosen, Phys. Rev. \textbf{47}, 777 (1935).
\bibitem{bell} S. Bell, Physics \textbf{1}, 195 (1964).
\bibitem{chsh} J. F. Clauser, M. A. Horne, A. Shimony, and R. A. Holt, \prl \textbf{23}, 880 (1969).
\bibitem{ch74} J. F. Clauser and M. A. Horne, Phys. Rev. D \textbf{10}, 526 (1974).
\bibitem{freedman} S. J. Freedman and J. F. Clauser, Phys. Rev. Lett. \textbf{28}, 938 (1972);
\bibitem{aspect} A. Aspect, P. Grangier, and G. Roger, Phys. Rev. Lett. \textbf{47}, 460 (1981);
A. Aspect, P. Grangier, and G. Roger, \prl \textbf{49}, 91 (1982);
A. Aspect, J. Dalibard, and G. Roger, \prl \textbf{49}, 1804 (1982).
\bibitem{shih} Y. H. Shih and C. O. Alley, \prl \textbf{61}, 2921 (1988);
\bibitem{ou} Z. Y. Ou and L. Mandel, \prl \textbf{61}, 50 (1988);
Z. Y. Ou, S. F. Pereira, H. J. Kimble, and K. C. Peng, \prl \textbf{68}, 3663 (1992).
\bibitem{rarity} J. G. Rarity and P. R. Tapster, Phys. Rev. Lett. \textbf{64}, 2495
(1990).
\bibitem{tittel} W. Tittel, J. Brendel, H. Zbinden, and N. Gisin, Phys. Rev.
Lett. \textbf{81}, 3563 (1998).
\bibitem{weihs} G. Weihs, T. Jennewein, C. Simon, H. Weinfurter, and A. Zeilinger, \prl \textbf{81}, 5039 (1998).
\bibitem{kwiat} P. G. Kwiat, E. Waks, A. G. White, I. Appelbaum, and P. H. Eberhard, \pra \textbf{60}, R773 (1999).
\bibitem{rowe} M. Rowe, D. Kielpinski, V. Meyer, C. A. Sackett, W. M. Itano, C. Monroe, and D. J. Wineland, Nature \textbf{409}, 791 (2001).
\bibitem{moehring} D. L. Moehring, M. J. Madsen, B. B. Blinov, and C.
Monroe, Phys. Rev. Lett. \textbf{93}, 090410 (2004).
\bibitem{matsukevich} D. N. Matsukevich, P. Maunz, D. L. Moehring, S. Olmschenk, and C. Monroe, Phys. Rev. Lett. \textbf{100}, 150404 (2008).
\bibitem{pearle} P. M. Pearle, Phys. Rev. D \textbf{2}, 1418 (1970).
\bibitem{lo} T. K. Lo and A. Shimony, \pra \textbf{23}, 3003 (1981);
\bibitem{garg} A. Garg and N. D. Mermin, \prd \textbf{35}, 3831 (1987).
\bibitem{santos} E. Santos, Phys. Rev. A \textbf{46}, 3646 (1992).
\bibitem{eber} P. H. Eberhard, Phys. Rev. A \textbf{47}, R747 (1993).
\bibitem{gisin} N. Gisin and B. Gisin, Phys. Lett. A \textbf{260}, 323 (1999).
\bibitem{fair} J. F. Clauser and A. Shimony, Rep. Prog. Phys. \textbf{41}, 1881 (1978).
\bibitem{shaf} A. Shafiee and M. Golshani, J. Mod. Opt. \textbf{52}, 1405 (2005).
\bibitem{tsirelson} B. S. Tsirelson, Lett. Math. Phys. \textbf{4}, 93 (1980).
\bibitem{massar} S. Massar, S. Pironio, J. Roland, and B. Gisin, Phys. Rev.
A \textbf{66}, 052112 (2002).
\bibitem{durt} T. Durt, D. Kaszlikowski, and M. \.{Z}ukowski, Phys. Rev. A
\textbf{64}, 024101 (2001).
\bibitem{genovese} M. Genovese, Phys. Rev. A \textbf{71}, 052314 (2005).
\bibitem{brun} N. Brunner, N. Gisin, V. Scarani, and C. Simon, \prl \textbf{98}, 220403 (2007).
\bibitem{Bennett} C. H. Bennett, H. J. Bernstein, S. Popescu, and B. Schumacher,
Phys. Rev. A \textbf{53}, 2046 (1996).
\bibitem{Kwiat01} P. G. Kwiat, S. Barraza-Lopez, A. Stefanov, and N. Gisin,
Nature \textbf{409}, 1014 (2001).
\bibitem{cab} A. Cabello, Phys. Rev. Lett. \textbf{88}, 060403 (2002);
A. Cabello, \pra \textbf{66}, 042114 (2002).
\bibitem{anne} A. Broadbent, H. A. Carteret, A. A. M\'ethot,
and J. Walgate, New Journal of Physics \textbf{8}, 302 (2006).
\bibitem{SVD} D. C. Lay, {\it Linear algebra and its applications}, (Addison-Wesley, Reading, 1996).
\bibitem{tera05} H. Terashima and M. Ueda, Int. J. Quantum Inform. \textbf{3}, 633 (2005).
\bibitem{bergou00} J. A. Bergou, M. Hillery, and Y. Sun, J. Mod. Opt. \textbf{47}, 487 (2000).
\bibitem{Kwiat95} P. G. Kwiat, K. Mattle, H. Weinfurter, A. Zeilinger,
A. V. Sergienko, and Y. Shih, \prl {\bf 75}, 4337 (1995).
\bibitem{KLM} E. Knill, R. Laflamme and G. J. Milburn, Nature {\bf 409}, 46 (2001).
\bibitem{ekert} A. K. Ekert, \prl \textbf{67}, 661 (1991).
\bibitem{shor} P. W. Shor and J. Preskill, \prl \textbf{85}, 441 (2000).
\bibitem{barrett} J. Barrett, L. Hardy, and A. Kent, Phys. Rev. Lett. \textbf{95},
010503 (2005).
\bibitem{acin} A. Ac\'{i}n, N. Gisin, and L. Masanes, \prl \textbf{97}, 120405 (2006);
A. Ac\'{i}n, N. Brunner, N. Gisin, S. Massar, S. Pironio, and V. Scarani, Phys. Rev. Lett. \textbf{98}, 230501 (2007).
\bibitem{pironio} S. Pironio, A. Ac\'{i}n, N. Brunner, N. Gisin, S. Massar, and V. Scarani, New J. Phys. \textbf{11}, 045021 (2009).
\bibitem{aden} G. Adenier and A. Yu. Khrennikov, J. Phys. B \textbf{40}, 131 (2007).
\bibitem{pawlowski} M. Paw{\l}owski, K. Horodecki, P. Horodecki, and
R. Horodecki, arXiv:0902.2162.

\end{thebibliography}
\end{document}